\documentclass[11pt]{article}

\usepackage{fullpage}
\usepackage{amsmath}
\usepackage{xspace}
\usepackage{amssymb}

\newtheorem{fact}{Fact}    
\newtheorem{definition}{Definition} 
\newtheorem{theorem}{Theorem}    
\newtheorem{corollary}{Corollary}    
\newcommand{\qed}{\hfill{$\rule{6pt}{6pt}$}} 
\newenvironment{proof}{\noindent{\bf Proof}:}{\qed}
\newenvironment{proofof}[1]{\noindent{\bf Proof of #1:}}{\qed}

\newcommand{\defeq}{\stackrel{\mathsf{def}}{=}}
\newcommand{\ket}[1]{| #1 \rangle}
\newcommand{\bra}[1]{\langle #1 |}
\newcommand{\ketbra}[1]{| #1 \rangle \langle #1 |}

\newcommand{\tn}[1]{\left\| #1 \right\|_{{\mathrm{tr}}}}

\newcommand{\cC}{{\mathcal C}}
\newcommand{\R}{{\mathbb R}}
\newcommand{\C}{{\mathbb C}}
\newcommand{\EQ}{{\mathrm{EQ}}}
\newcommand{\E}{{\mathbb{E}}}

\newcommand{\Tr}{{\mathsf{Tr}}}

\newcommand{\cH}{{\cal H}}
\newcommand{\cK}{{\cal K}}

\newcommand{\cP}{{\cal P}}

\newcommand{\cX}{{\cal X}}
\newcommand{\cY}{{\cal Y}}
\newcommand{\cZ}{{\cal Z}}
\newcommand{\sQ}{{\mathsf{Q}}}
\newcommand{\sR}{{\mathsf{R}}}
\newcommand{\sD}{{\mathsf{D}}}
\newcommand{\cS}{{\mathcal{S}}}

\newcommand{\smp}{{\sf {SMP}}}
\newcommand{\pub}{{\sf {pub}}}
\newcommand{\ent}{{\sf {ent}}}
\newcommand{\tpriv}{{\sf {\widetilde{priv}}}}
\newcommand{\priv}{{\sf {priv}}}

\newcommand{\ti}{\tilde{i}}
\newcommand{\alice}{\sf {Alice}}
\newcommand{\bob}{\sf {Bob}}
\newcommand{\referee}{\sf {Referee}}
\newcommand{\suppress}[1]{}
\newcommand{\cq}{\mathsf {c-q}}

\title{\bf New Results in the Simultaneous Message Passing Model}
\author{
Rahul Jain \thanks{Centre for Quantum Technologies and Department of Computer Science, National University of Singapore.
3 Science Drive 2, Singapore 117543. Email: {\sf rahul@comp.nus.edu.sg}. This work is supported by the National Research Foundation and Ministry of Education, Singapore. }
\\ National University of Singapore
\and
Hartmut Klauck \thanks{Centre for Quantum Technologies, National University of Singapore. 3 Science Drive 2, Singapore 117543. Email: {\sf hklauck@gmail.com}. This work is supported by the National Research Foundation and Ministry of Education, Singapore.}
\\ National University of Singapore
}

\date{}

\begin{document}

\maketitle
\begin{abstract}
Consider the following {\em Simultaneous Message Passing} ($\smp$) model for
computing a relation $f \subseteq \cX \times \cY \times \cZ$. In this
model $\alice$, on input $x \in \cX$ and $\bob$, on input $y \in \cY$,
send one message each to a third party $\referee$ who then outputs a
$z \in \cZ$ such that $(x,y,z) \in f$. We first show
optimal {\em Direct sum} results for all relations $f$ in this model, both in the quantum and
classical settings, in the situation where we allow shared resources
(shared entanglement in quantum protocols and public coins in
classical protocols) between $\alice$ and $\referee$ and $\bob$ and
$\referee$ and no shared resource between $\alice$ and $\bob$.
This implies that, in this model, the communication required to compute
$k$ simultaneous instances of $f$, with constant success overall, is at least $k$-times the communication
required to compute one instance with constant success.

This in particular implies an earlier Direct sum result, shown by Chakrabarti, Shi, Wirth and Yao~\cite{ChakrabartiSWY01}
for the Equality function (and a class of other so-called robust functions), in the classical $\smp$ model with no shared resources
between any parties.

Furthermore we investigate the gap between the $\smp$ model and the one-way model in communication complexity and exhibit
a partial function that is exponentially more expensive in the former if quantum communication with entanglement
is allowed, compared to the latter even in the deterministic case.
\end{abstract}


\newpage

\section{Introduction}

\subsection{The Direct sum problem}
The Direct sum question  asks if computing $k$ instances of a given function or relation together, with constant success overall, requires $k$-times the resources required for computing one instance, with constant success. It is a widely studied question and its resolution in some settings lead to important consequences.  Karchmer, Raz, and Wigderson~\cite{KarchmerRW95} show that a Direct sum result for deterministic communication complexity of certain relations would probably imply ${\mathsf{NC}}^1 \not=
{\mathsf{NC}}^2$. Bar-Yossef, Jayram, Kumar, and
Sivakumar~\cite{Bar-YossefJKS04} use Direct sum results to prove space
lower bounds in the datastream model~\cite{Bar-YossefJKS04}.
P\v{a}tra\c{s}cu and Thorup~\cite{PatrascuT06} use Direct sum type
results to prove stronger lower bounds for approximate near-neighbor
(ANN) search in the cell probe model. Work on the Direct sum property
has also inspired earlier lower bounds for ANN due to Chakrabarti and
Regev~\cite{ChakrabartiR04}.

Although they seem highly plausible, it is well-known that
Direct sum results fail to hold for some modes of communication. For example, testing the equality of~$k = \log n$
pairs of~$n$-bit strings with a constant-error private-coin
communication protocol has complexity~$O(k\log k + \log n) =
O(\log n \log\log n)$ (see, e.g., \cite[Example~4.3,
page~43]{KushilevitzN97}), where we might expect a complexity
of~$\Omega(k\log n) = \Omega(\log^2 n)$.

We consider this question in certain Simultaneous Message Passing models of communication complexity and answer in the affirmative.
To be more precise, let $f \subseteq \cX \times \cY \times \cZ$ be a relation, where $\cX, \cY, \cZ$ are finite sets.  For a positive integer $k$, let's define the {\em $k$-fold product\/} of $f$, $f^{\otimes k} \subseteq \cX^{k}
\times \cY^{k} \times \cZ^{k} $ as $f^{\otimes k}\defeq \{ (x_1,
\ldots x_k, y_1, \ldots, y_k, z_1, \ldots, z_k): \forall i \in [k],
(x_i,y_i,z_i) \in f \}$. This relation captures~$k$ independent
instances of the relation~$f$.  Details of the $\smp$ models we consider and the definitions of corresponding communication complexities appear in Sec.~\ref{sec:prelimcomm}.
We show the following result.

\begin{theorem}[Direct sum]
\label{thm:direct}
Let $f \subseteq \cX \times \cY \times \cZ$ be a relation. Let $k$ be a positive integer. Let $\epsilon, \delta \in (0,1/4)$. Then,
\begin{enumerate}
\item  $\sQ^{\|, \tpriv}_{\epsilon} (f^{\otimes k}) \quad \geq \quad \Omega(k \cdot \delta^3 \cdot \sQ^{\|, \tpriv}_{\epsilon + \delta} (f)) \enspace .$
\item  $\sR^{\|, \tpriv}_{\epsilon} (f^{\otimes k}) \quad \geq \quad \Omega(k \cdot \delta^3 \cdot \sR^{\|, \tpriv}_{\epsilon + \delta}(f)) \enspace .$
\end{enumerate}
\end{theorem}

Here $\sQ^{\|, \tpriv}_{\epsilon} (f)$ denotes the communication complexity of a relation $f$ in the quantum simultaneous message passing model with no shared resources between $\alice$ and $\bob$, but shared entanglement between $\alice$ and $\referee$ resp.~$\bob$ and $\referee$. Similarly, for  $\sR^{\|, \tpriv}_{\epsilon} (f)$ $\alice$ and $\bob$ share no resources, but $\alice$ and $\referee$ have shared access to a source of random bits coin (not seen by $\bob$), $\bob$ and $\referee$ access to a different source (not seen by $\alice$).

Using standard arguments due to Newman~\cite{Newman91} one can show that for any relation $f \subseteq \cX \times \cY \times \cZ$,
$$ \sR^{\|, \tpriv}(f)) \quad \geq \quad  \Omega( \sR^{\|, \priv}(f) - O(\log |\cX|  + \log |\cY|) ) \enspace .$$

Hence we obtain the following corollary of Thm.~\ref{thm:direct}:
\begin{corollary}
\label{corr:directclass}
Let $f \subseteq \cX \times \cY \times \cZ$ be a relation and let $k$ be a positive integer.  Then,
 $$\sR^{\|, \priv} (f^{\otimes k}) \quad \geq \quad \sR^{\|, \tpriv}(f^{\otimes k}) \quad \geq \quad  \Omega(k \cdot \sR^{\|, \tpriv}(f)) \quad \geq \quad  \Omega(k  \cdot (\sR^{\|, \priv}(f) - O(\log |\cX| + \log |\cY|))) \enspace .$$
\end{corollary}

Note that a similar result to Newman's is unknown for the quantum model (and probably does not hold), so we do not get a corresponding tight Direct sum result in the quantum case for the  $\smp$ model where no entanglement is shared between any pair among $\alice$/$\bob$/$\referee$.


\subsection{One-way vs.~simultaneous messages}

It is clear that one-way protocols, in which either $\alice$ or $\bob$ sends one message to the other
player, who then outputs the result, can easily simulate simultaneous message passing protocols, hence if we
denote the maximum of the one-way complexities (over the choice of the player sending the message)
by $\sR^1(f)$ (we will use similar notations for the other modes of communication), we immediately get conclusions like $\sR^1_\epsilon(f)\leq \sR^{\|}_{\epsilon}(f)$. But how much smaller can the one-way communication be compared to the  $\smp$-complexity?

For deterministic complexity it is easy to see that $\sD^1(f)=\Theta(\sD^{\|}(f))$ for all total functions $f$. Bar-Yossef et al. \cite{BarYossefJKS02} exhibit a total function $g$ for which $\sR^1(g)=O(\log n)$, while $\sR^\|(g)=\Omega(\sqrt n)$.

We first generalize this result to the quantum case, showing that $\sQ^{\|, \pub} (g)=\Omega(\sqrt n)$ as well. Just like in \cite{BarYossefJKS02} the lower bound is based on giving a lower bound for the Generalized Addressing Function of \cite{BabaiGKL03}. In fact all known lower bounds for this function are based on a certain subfunction, for which $\Omega(\sqrt n)$ is tight, whereas the exact complexity of the Generalized Addressing Function is open, see \cite{AmbainisL00} for the best known upper bound.
However, the proof of the above lower bound fails, when we allow entanglement between $\alice$ and $\bob$. So we consider a different partial function $f$ which has the desired behavior even if we allow arbitrary tripartite entanglement.

\begin{theorem} There is a partial Boolean function $f$ on $n$ inputs such that
$\sD^1(f)\leq\log n$, while $\sQ^{\|, \ent}(f)\geq\Omega(\sqrt n)$.
\end{theorem}

 Note that a similar result cannot be true for a total function (the function $g$ above only has a randomized upper bound for one-way protocols).

\subsection{Previous work on Direct sum}

Babai and Kimmel~\cite{BabaiK97}, following arguments as in Newman~\cite{Newman91}, show the following.
\begin{fact}[\cite{BabaiK97}] For a relation $f \subseteq \cX \times \cY \times \cZ$, let $\sD^{\|}(f)$ represent the deterministic communication complexity for computing $f$ in the $\smp$ model. Then, $\sR^{\|, \priv}(f) = \Omega(\sqrt{\sD^{\|}(f)}) \enspace .$
\end{fact}

The Direct sum result for $\sD^{\|}(f)$ is easy to show and hence one can derive the following Direct sum result for $\sR^{\|, \priv}(f)$\footnote{Note that this result is weaker than our result Corr.~\ref{corr:directclass}, whenever $\sR^{\|, \priv}(f) = \Omega(\log |\cX| + \log |\cY|)$.}.
\begin{fact}[Implicit from~\cite{BabaiK97}] Let $f \subseteq \cX \times \cY \times \cZ$ be a relation and let $k$ be a positive integer. Then,
$$\sR^{\|, \priv}(f^{\otimes k}) = \Omega(\sqrt{\sD^{\|}(f^{\otimes k})}) = \Omega(\sqrt{k \cdot \sD^{\|}(f) }) = \Omega(\sqrt{k \cdot \sR^{\|, \priv}(f) }) \enspace .$$
\end{fact}

Chakrabarti, Shi, Wirth and Yao~\cite{ChakrabartiSWY01} consider the Direct sum problem in the
private coins $\smp$ model and show the following result.
For a function $f:\{0,1\}^n \times \{0,1\}^n \rightarrow \{0,1\}$, let
$\tilde{\sR}^{{\|, \priv}}(f) \defeq \min_{S} \sR^{\|, \priv} (f|_{S\times S})$,
where $S$ ranges
over all subsets of $\{0,1\}^n$ of size at least $(\frac{2}{3}) 2^n$ and $f|_{S\times S}$ denotes the function $f$ restricted to inputs $x,y$ both from the set $S$. It is easily seen that $\tilde{\sR}^{{\|, \priv}}(f) \leq \sR^{{\|, \priv}}(f).$
\begin{fact}[\cite{ChakrabartiSWY01}] Let $k$ be a positive integer. Then,
$$\sR^{\|, \priv}(f^{\otimes k})= \Omega(k \cdot (\tilde{\sR}^{\|, \priv}(f)-O(\log n))) \enspace .$$
\end{fact}

For the Equality function $\EQ_n : \{0,1\}^n \times \{0,1\}^n \rightarrow \{0,1\}$, in which the $\referee$ outputs $1$ iff the inputs of $\alice$ and $\bob$ are equal, it can easily be seen that $\tilde{\sR}^{\|, \priv}(\EQ_n) = \Theta(\sR^{\|, \priv}(\EQ_n)) $.  Hence the above result provides an optimal Direct sum result for $\EQ_n$.

In the $\smp$ models in which $\alice$ and $\bob$ share public coins, optimal Direct sum results have been shown earlier by Jain, Radhakrishnan and Sen~\cite{JainRS05}.
\begin{fact}[Direct sum~\cite{JainRS05}]
Let $f \subseteq \cX \times \cY \times \cZ$ be a relation. Let $k$ be a positive integer. Let
$\epsilon, \delta \in (0,1/4)$, then
\begin{enumerate}
\item $\sQ^{\|, \pub}_\epsilon (f^{\otimes k}) \quad \geq \quad
 \Omega\left(k \cdot \delta^3 \cdot \sQ^{\|,\pub}_{\epsilon + \delta} (f)\right).$
\item $\sR^{\|, \pub}_\epsilon (f^{\otimes k}) \quad \geq \quad
 \Omega\left(k \cdot \delta^3 \cdot \sR^{\|,\pub}_{\epsilon+ \delta} (f)\right).$
\end{enumerate}
\end{fact}

Note that for the Equality function there is an exponential gap between the classical randomized public coin $\smp$- and the
private public coin $\smp$-complexity (see e.g.~\cite{BabaiK97}). A similar exponential gap is known for a relation in the quantum model \cite{GavinskyKRW06}, i.e.~there is a relation $r$ with $\sR^{\|, \pub}(r)\leq\log n$ and $\sQ^{\|, \tpriv}(r)\geq\Omega(n^{1/3})$ (the paper states only a $\sQ^{\|, \priv}$ bound, but the proof can be extended easily). Hence the previous results in the public model do not imply ours, and in particular any approach using arguments about distributional communication complexity is not possible to establish Thm.~\ref{thm:direct}, due to the inherent connection to public coin complexity. Furthermore we believe our proof is simpler than the proofs of the
classical Direct sum result by \cite{ChakrabartiSWY01} and the above result. This is achieved by viewing the communications from $\alice$/$\bob$ as a communication
channel in the Shannon sense (i.e.~not fixing the underlying probability distributions of the maps from inputs to messages), which allows for worst case message compression as opposed to the previous average case arguments.

\subsection{Organization} In the next section we present the necessary definitions and facts that are subsequently used in our proofs. In Sec.~\ref{sec:direct} we present the proofs of our Direct sum results.  Sec.~\ref{sec:comparing} contains the results comparing one-way- to $\smp$-complexity. We conclude in Sec.~\ref{sec:conc} with some open problems. For completeness, in Sec.~\ref{sec:proofs}, we present the proofs of the earlier known facts that we use in this work.

\section{Preliminaries}

\subsection{Information theory} For an operator $A$, its {\em trace norm} is defined to be $\tn{A} \defeq \Tr \sqrt{A^\dagger A}$. We use the {\em bra-ket} notation in which a vector is represented as $\ket{\phi}$ and its adjoint is represented as $\bra{\phi}$. A {\em quantum state} is a positive semi definite trace one operator. A {\em pure state} is a quantum state of rank one and is often represented by its sole eigenvector with non-zero eigenvalue. For a quantum state $\rho$ in Hilbert space $\cH$, a pure state $\ket{\phi} \in \cH \otimes \cK$ is called its {\em purification} if $\Tr_\cK \ket{\phi}\bra{\phi} = \rho$. For a quantum state $\rho$, its {\em von-Neumann
 entropy} is defined as $S(\rho) \defeq \sum_i -\lambda_i \log
\lambda_i$, where $\lambda_i$s represent the various eigenvalues of
$\rho$.  It is easily seen that for an $l$ qubit quantum system $A$ with state $\rho_A$, $S(A)\defeq  S(\rho_A) \leq l$.  For systems $A,B$ their {\em mutual information} is defined as $I(A:B)
\defeq S(A) + S(B) - S(AB)$.  Given
quantum states $\rho,\sigma$, their {\em relative entropy} is defined
as $S(\rho \| \sigma) \defeq \Tr \rho (\log \rho - \log \sigma)$. For
a joint classical-quantum system $XM$, where $X$ is a classical random
variable, let state of $M | (X=x)$ be $\rho_x$. Let $\rho \defeq \E_{x \leftarrow
X}[\rho_x]$. Then we have an alternate characterization of $I(X:M)$ as follows:
\begin{equation} \label{eq:altchar} I(X:M)= \E_{x \leftarrow X}[S(\rho_x \| \rho)] \enspace. \end{equation}

For classical random variables the analogous definitions and facts hold {\em mutatis mutandis}.

\subsection{Communication complexity}
\label{sec:prelimcomm}
\subsubsection*{Quantum communication complexity}  In a Simultaneous Message Passing ($\smp$) quantum communication protocol $\cP$ for computing a
relation $f \subseteq \cX \times \cY \times \cZ$,  $\alice$ and
$\bob$ get inputs $x \in \cX$ and $y \in
\cY$ respectively. They each send a message to a third party called {\sf
Referee}. The $\referee$ then outputs a $z \in \cZ$ such that
$(x,y,z) \in f$. The internal computations and messages send by the parties can be quantum. On any input pair $(x,y)$, the protocol can err with a small probability. The relations we consider are always total in the sense that for
every $(x,y) \in \cX \times \cY$, there is at least one $z \in
\cZ$, such that $(x,y,z) \in f$. There are four models of quantum $\smp$ protocols that we consider. Given $\epsilon \in (0,1/2)$, the communication
complexity in any given model is defined to be the communication of the
best $\smp$ protocol in that model, with error at most $\epsilon$ on all inputs.  In the first model there is no shared resource between any of the parties and the communication
complexity in this model is denoted by $\sQ^{\|, \priv}_\epsilon (f)$. In the second model we allow prior
entanglement to be shared between $\alice$ and $\referee$,
$\bob$ and $\referee$, but no shared resource between $\alice$ and $\bob$. The entangled state for $\alice$ and $\referee$ is independent of the entangled state for $\bob$ and $\referee$.
The communication complexity in this model is denoted by $\sQ^{\|, \tpriv}_\epsilon (f)$.
In the third model, we allow prior entanglement to be shared between $\alice$ and $\referee$,
$\bob$ and $\referee$, and public coins to be shared between $\alice$ and $\bob$. The communication complexity in this model is  denoted by $\sQ^{\|, \pub}_\epsilon (f)$. Finally, in Sec.~\ref{sec:comparing} we will also consider the model, in which $\alice$, $\bob$, and $\referee$ share an arbitrary entangled tripartite state and the communication complexity in this model is  denoted by $\sQ^{\|, \ent}_\epsilon (f)$.
Whenever the error parameter $\epsilon$ is not specified it is assumed to be $1/3$.

\subsubsection*{Classical communication complexity} In the classical models, the internal computations by the parties and the messages sent are classical.
Similar to the quantum case, we consider three models of classical $\smp$ protocols. In the first model, there is no shared resource between any of the parties and the communication complexity is denoted by $\sR^{\|, \priv}_\epsilon(f)$.
In the second model, we let the public coins to be
shared between $\alice$  and $\referee$, $\bob$ and $\referee$ and no shared resource between $\alice$ and $\bob$. The communication complexity in this model is denoted by $\sR^{\|, \tpriv}_\epsilon(f)$.  In the third model, we let public coins to be shared between $\alice$  and $\referee$, $\bob$ and $\referee$ and
between $\alice$ and $\bob$. The communication complexity in this
model is denoted by $\sR^{\|, \pub}_\epsilon(f)$. As before whenever error parameter $\epsilon$ is not specified it is assumed to be $1/3$.

\subsection{Useful facts}
Here we present some known facts that will subsequently be
 useful in our proofs. We provide proofs for some of them
in Sec.~\ref{sec:proofs} for completeness.
We state them here in the quantum case. In the classical case, these hold {\em mutatis mutandis} by replacing
quantum states by probability distributions and we avoid making
explicit statements and proofs.

The following fact is probably folklore and appears among other places for example in~\cite{JainRS05}.
\begin{fact}
\label{fact:lowinfent}
Let $XMN$ be a tri-partite system with $X$ being a classical system. If $I(X:M)=0$ then $I(X:MN) \leq 2S(N)$.
\end{fact}
Let $\cX$ be a finite set and let $\cS$ be the set of all quantum
states.  A classical-quantum ($\cq$) channel $E$ is a map from
$\cX$ to $\cS$.  All the channels we consider will be $\cq$ channels and we will avoid mentioning $\cq$ explicitly from now on.
For a probability distribution $\mu$ over $\cX$, let $E_\mu$ be the bipartite state $\E_{x \leftarrow \mu}[\ketbra{x}
\otimes E(x)]$. Let $I(E_\mu)$ be the mutual
information  between the two systems in $E_\mu$.  The channel capacity of such a channel is defined as follows.
\begin{definition}[Channel capacity]
Channel capacity of the channel $E: \cX \mapsto \cS$ is
defined as $C(E) \defeq \max_{\mu} I(E_\mu)$.
\end{definition}
A {\em derived channel} is defined as follows.
\begin{definition}[Derived channel]
Let $\cX$ and $\cY$ be finite sets. Let $E: \cX \times \cY \rightarrow \cS$ be a channel.  For a collection $\{\mu_x : x \in \cX \}$, where each
$\mu_x$ is a probability distribution on $\cY$, let $F: \cX \rightarrow \cS$ be a channel given by $ F(x) \defeq
\E_{y \leftarrow \mu_x}[E(x,y)]$. Such  a channel $F$ is referred to as an
$E$-derived channel on $\cX$. Similarly we can define $E$-derived
channels on $\cY$ using collections of probability distributions on $\cX$.
\end{definition}
We will need the following  result from Jain~\cite{Jain05}.
\begin{fact}[Super-additivity~\cite{Jain05}] \label{fact:supadd}
Let $k$ be a positive integer. Let $\cX_1, \cX_2, \ldots, \cX_k$ be finite sets. Let  $E:
\cX_1 \times \cX_2 \ldots \times \cX_k \rightarrow \cS$ be a channel. For $i \in [k]$, let $\cC_{i}$ be the set of
all $E$-derived channels on $\cX_i$. Then,
$$C(E) \quad \geq \quad \sum_{i=1}^k \min_{F_i \in \cC_i} C(F_i) \enspace .   $$
\end{fact}

We will also use the following result from Jain~\cite{Jain06}. An alternate proof of this fact for the special case of classical channels, can be found in~\cite{HarshaJMR07}.
\begin{fact}[\cite{Jain06}]
\label{fact:SleqC}
Let $E : \cX \rightarrow \cS$ be a channel. There exists a
quantum state $\tau$ such that $$\forall x \in \cX, \quad S(E(x) \| \tau ) \quad \leq \quad C(E) \enspace .$$
\end{fact}

The above fact allows worst case message compression when the channel
capacity is small: given $\tau$ we can reconstruct {\it any} $E(x)$ using the following compression result implicit in~\cite{JainRS05} (stated slightly differently there).

\begin{fact}[Compression~\cite{JainRS05}]
\label{fact:compress} Let $\alice$ and $\referee$ share several copies of a bi-partite pure state $\ket{\phi}$ between  them, such that the marginal of $\ket{\phi}$ on $\referee$'s part is $\tau$.
For any state $\rho$ and for any $\delta > 0$, $\alice$ can measure her part of the states and send $O(\frac{1}{\delta^3} \cdot S(\rho \| \tau))$ bits to $\referee$, enabling $\referee$ to pick state $\rho'$ with him such that $\tn{\rho - \rho'} \leq \delta$.
\end{fact}

We explicitly state the classical version of the above result for clarity.
\begin{fact}[Compression~\cite{JainRS05}]
\label{fact:compressclass} Let $\alice$ and $\referee$ share public coins distributed according to $Q$.
For any distribution  $P$ and for any $\delta > 0$, $\alice$ can send $O(\frac{1}{\delta^2} \cdot S(P \| Q))$ bits to $\referee$, at the end of which $\referee$ can sample from a distribution $P'$ such that $\|P - P'\| \leq \delta$.
\end{fact}

We will use the following relation between relative entropy and trace distance from \cite{KlauckNTZ07}.
\begin{fact}
\label{fact:entropytrace} For density matrices $\rho,\sigma:$
\[\tn{\rho-\sigma}\leq \sqrt2S(\rho\|\sigma)^{1/2}.\]
\end{fact}

Finally, we need the quantum random access code bound due to Nayak \cite{Nayak99} (here also
stated for the case where entanglement is allowed).

\begin{fact}
\label{fact:qrac} Assume $\alice$ receives a uniformly random string $x\in\{0,1\}^n$ and $\bob$ a uniformly random
index $i\in\{1,\ldots,n\}$. $\alice$ and $\bob$ may share entanglement, and $\alice$ sends one message to $\bob$, which allows him to decode $x_i$ with probability $1-\epsilon$ (averaged over the inputs). Then $\alice$'s message needs to have $(1-H(\epsilon))n/2$ qubits, where $H$ denotes the binary entropy function. Without entanglement the bound is $(1-H(\epsilon))n$.
\end{fact}

The above result is essentially a lower bound in the quantum one-way communication complexity model for a function known as the Index function. Alternatively we will refer to $\alice$'s message as the random access code of the strings $x$.
\section{Direct sum}
\label{sec:direct}

We restate and subsequently prove our main result about Direct sum.
\begin{theorem}[Direct sum]
Let $f \subseteq \cX \times \cY \times \cZ$ be a relation. Let $k$ be a positive integer. Let $\epsilon, \delta \in (0,1/4)$. Then,
\begin{enumerate}
\item \label{it:quant} $\sQ^{\|, \tpriv}_{\epsilon} (f^{\otimes k}) \quad \geq \quad \Omega(k \cdot \delta^3 \cdot \sQ^{\|, \tpriv}_{\epsilon + \delta} (f)) \enspace .$
\item \label{it:class} $\sR^{\|, \tpriv}_{\epsilon} (f^{\otimes k}) \quad \geq \quad \Omega(k \cdot \delta^2 \cdot \sR^{\|, \tpriv}_{\epsilon + \delta}(f)) \enspace .$
\end{enumerate}
\end{theorem}

\begin{proof}
We state the proof of part~\ref{it:quant} above. The proof of
part~\ref{it:class} follows very similarly by using the classical
versions of the facts used.

Let $c \defeq \sQ^{\|, \tpriv}_\epsilon (f^{\otimes k})$. Let $\cP$ be an $\smp$ protocol for $f^{\otimes k}$ with communication $c$ and its error on all inputs being at most $\epsilon$.
Let $\rho_x$ be the combined state of the qubits received by $\referee$  from $\alice$, when  $\alice$'s input is $x$, and $\referee$'s part of the shared entangled state with $\alice$. Similarly let
 $\sigma_y$ be the combined state of the qubits received by $\referee$  from $\bob$, when $\bob$'s input is $y$, and $\referee$'s part of the shared entangled state with $\bob$.  Let $\cS$ be the set of all quantum states.
Let $A:  \cX \rightarrow \cS$ be a channel given by $A(x) \defeq
\rho_x$ and let $B:  \cY \rightarrow \cS$ be a channel given by $B(y)
\defeq \sigma_y$.  Using Fact~\ref{fact:lowinfent} and the fact that
for an $l$ qubit quantum system $M$, $S(M) \leq l$, it can be seen
that $C(A) \leq 2c$ and $C(B) \leq 2c$. From Fact~\ref{fact:supadd} and
using Markov's inequality, we have that there exists a coordinate $i
\in [k]$ and an $A$-derived channel $A_i$ on the input on the $i$-th
coordinate and a $B$-derived channel $B_i$ on the input on the $i$-th
coordinate, such that $C(A_i) \leq \frac{4c}{k}$ and $C(B_i) \leq
\frac{4c}{k}$.

We will now present a protocol $\cP'$ for $f$.  In $\cP'$, $\alice$ on
input $x$, sends state $A_i(x)$ to $\referee$. Similarly $\bob$ on
input $y$, sends state $B_i(y)$ to $\referee$. $\referee$ performs the same actions
as in $\cP$ and outputs the result corresponding to the $i$-th coordinate. It can be seen that the error in $\cP'$,
on any input pair $(x,y)$ is bounded by $\epsilon$.

Now we present the final protocol $\cP''$. Let $\tau_a$ be the state obtained from Fact~\ref{fact:SleqC}
such that $\forall x \in \cX, \quad S(A_i(x)|| \tau_a ) \leq C(A_i)$.   Similarly let $\tau_b$ be the state
obtained from Fact~\ref{fact:SleqC} such that $\forall y \in \cY, \quad S(B_i(y)|| \tau_b ) \leq C(B_i)$.
Let $\ket{\phi_a}$ be a purification of $\tau_a$ and let $\ket{\phi_b}$ be a purification of $\tau_b$.  $\alice$ and
$\referee$ share several copies of $\ket{\phi_a}$ as shared
entanglement in $\cP''$. $\bob$ and $\referee$ share several
copies of $\ket{\phi_b}$ as shared entanglement in $\cP''$. $\alice$,
on receiving input $x$, using Fact~\ref{fact:compress} sends
$O(\frac{1}{\delta^3} \cdot S(A_i(x) \| \tau_a))$ bits to $\referee$
at the end of which $\referee$ has a state $\rho_x'$ such that
$\tn{A_i(x) - \rho_x'} \leq \delta$. Similarly $\bob$, on receiving
input $y$, using Fact~\ref{fact:compress} sends $O(\frac{1}{\delta^3}
\cdot S(B_i(y) \| \tau_b))$ bits to $\referee$ at the end of which
$\referee$ has a state $\sigma_x'$ such that $\tn{A_i(x) - \sigma_x'}
\leq \delta$. It can be seen that the error of protocol $\cP''$ on any
input pair $(x,y)$ is bounded by $\epsilon + 2\delta$. Also the
communication for any input pair is bounded by $\frac{4c}{k
\delta^3}$. Hence we can conclude part~\ref{it:quant} from the definitions
of $\sQ^{\|, \tpriv}_{\epsilon} (f^{\otimes k})$ and $\sQ^{\|,
\tpriv}_{\epsilon + \delta} (f)$.
\end{proof}

\section{Comparing simultaneous messages and one-way communication}
\label{sec:comparing}

Recall that $\sD^1(f)$ denotes the maximum of the deterministic one-way communication complexities over $\alice$ and $\bob$ sending the message. It is easy to see that $\sD^1(f)=\Theta(\sD^{\|}(f))$ for all total functions $f$. Bar-Yossef et al. \cite{BarYossefJKS02} describe a total function $g$ for which $\sR^1(g)=O(\log n)$, while $\sR^\|(g)=\Omega(\sqrt n)$. This function is a variant of the Generalized Addressing Function investigated in \cite{BabaiGKL03}.

 For $g$ $\alice$ receives inputs $x\in\{0,1\}^n$ and $i\in\{1,\ldots,n\}$, $\bob$ $y\in\{0,1\}^n$ and $j\in\{1,\ldots,n\}$, and $g(x,i,y,j)=1\iff x=y \mbox{ and } x_{i\oplus j}=1$. The upper bound on $\sR^1(g)$ is straightforward and based on fingerprinting. For the lower bound one can restrict the inputs to $x=y$, and arrive at an equivalent of the 3-party number on the forehead Generalized Addressing Function from \cite{BabaiGKL03} over $Z_2^n$, for which the corresponding lower bound is $\Omega(\sqrt n)$. In fact this lower bound can be shown for the easier problem $h$ defined like $g$, except that $i$ and $j$ are strings of length $\log (n)/2$, and we are interested in the bit $x_k$ for which $k$ is the concatenation of $i$ and $j$. While for $h$ the resulting lower bound is obviously tight, the exact complexity of the Generalized Addressing Function remains open \cite{AmbainisL00}.

We will describe a partial function for which $\sD^1(f)\leq\log n$, while the quantum  $\smp$-complexity with entanglement is still $\Omega(\sqrt n)$. But first let us generalize the result of
\cite{BarYossefJKS02} to the quantum case. The lower bound builds on and simplifies the information theoretic part of the proof in \cite{BabaiGKL03}. In fact we simply reduce the problem to random access coding.

\begin{theorem}
$\sQ^{\|,\pub}(h)=\Omega(\sqrt n)$, while $\sR^1(h)=O(\log n)$.
\end{theorem}

\begin{proof}
We restrict the inputs to the set where $x=y$. For clarity let us first present a lower bound on $\sQ^{\|,\priv}(h)$. The plan is to construct a short quantum random access code from the messages in the protocol. For fixed $x$ $\alice$ is left with $\sqrt n$ different inputs $i$, similarly $\bob$ has only $\sqrt n$ different inputs $j$. Let the messages of $\alice$ be denoted by $\sigma_i$ and the messages of $\bob$ by $\rho_j$. We claim that the collection of all these messages forms a random access code for $x$. By the correctness of the protocol $\referee$ has a measurement that, applied to $\sigma_i\otimes\rho_j$ produces $x_{ij}$ with high probability for all $i,j$, which is exactly what we require. Hence all the $2\sqrt n$ messages together must have an average length of $(1-H(\epsilon))n$ (over the choice of $x$) via Fact~\ref{fact:qrac} to achieve success probability $1-\epsilon$, and consequently at least one message of the  $\smp$-protocol must have length $\Omega(\sqrt n)$.

To establish the same bound in the case $\alice$ and $\referee$ as well as $\bob$ and $\referee$ share entanglement, and $\alice$ and $\bob$ a classical public coin, note that we can produce a one-way protocol with entanglement for the Index function in the same way as above by composing the different messages of $\alice$ and $\bob$ (with $\referee$ holding the additional entanglement).
\end{proof}

The same lower bound obviously extends to the Generalized Addressing function over $Z_2^n$. It is easy to see that the proof can be generalized to the Generalized Addressing function over other groups and to the multiparty setting along the lines of the arguments in \cite{BabaiGKL03}.

Now note that the above proof fails if we allow entanglement between $\alice$ and $\bob$, since the messages $\sigma_i$ and $\rho_j$ will in general be entangled and so we cannot simply collect all of them while preserving the pairwise entanglement. We still conjecture the lower bound to hold for the quantum case with entanglement, but have not yet been able to show this. Instead we will construct a partial function (on $n^2$ inputs) for which $\sD^1(f)\leq\log n$ while every quantum  $\smp$ protocol needs communication $\Omega(n)$, even if $\alice$, $\bob$, and Referee share arbitrary tripartite entanglement. Note that such a result does not hold for total functions.

In fact the separation we seek is easily established for the following {\em relation} $s$: Let $\alice$ be given $x\in\{0,1\}^n$ and $i\in\{1,\ldots, n\}$, while $\bob$ gets $y\in\{0,1\}^n$ and $j\in\{1,\ldots, n\}$. Solving the relation requires us to output either $x_j$ or $y_i$ (and to indicate which). Clearly, $\sD^1(s)\leq\log n$. On the other hand a lower bound for the quantum $\smp$-model can be argued along the following lines: For each input one of the two allowed outputs must be made with probability at least $(1-\epsilon)/2$ (assuming error $\epsilon$). Hence under the uniform distribution on all inputs we are able to compute either $x_j$ or $y_i$ with probability $1/2-\epsilon/2$.
If we, say, can compute $x_j$ under the uniform distribution then we may toss a coin in case the protocol produces the other output. This leads to a simultaneous message protocol that computes the Index function with probability almost $3/4-\epsilon/2$. Hence the communication must be $\Omega(n)$, even with quantum messages and arbitrary entanglement, see Fact~\ref{fact:qrac}.

We now describe a partial Boolean function with the same behavior.
\begin{definition}
Let $\alice$ receive inputs $x\in\{0,1\}^n$ and $i\in\{1,\ldots, n\}$, while $\bob$ receives $n$ inputs $y_1,\ldots, y_n\in\{0,1\}^n$, and $j\in\{1,\ldots, n\}$. The promise is that $y_i=x$ and the desired function value is $f(x,i,y,j)=x_j$.
\end{definition}

Note that this function is essentially the Index function, but with enough side-information to allow it being computable by one-way protocols in both directions. Furthermore, this side-information is obfuscated in such a way as to make it useless in the  $\smp$-model.

\begin{theorem}
$\sD^1(f)\leq\log n$, while $\sQ^{\|, \ent}(f)\geq\Omega(n)$.
\end{theorem}

\begin{proof} For the upper bound note that there are deterministic  $\smp$-protocols, in which either $\alice$ or $\bob$ sends only $\log n$ bits, and the other player $n$ bits. These protocols can be easily simulated in the one-way model.

For the lower bound we show that if $\bob$ sends $\delta n$ qubits only and the error is $\epsilon$, then $\alice$ must send $(1-H(\epsilon+\sqrt\delta))n/2$ qubits. Hence for constant $\epsilon$, one of the messages has length $\Omega(n)$.

Assuming that $\bob$ sends $\delta n$ qubits only, we show that an $\smp$ protocol $\cP$ for $f$ (with worst case error $\epsilon$ on inputs satisfying the promise $y_i=x$) can be turned into an $\smp$ protocol $\cP'$ for the Index function. In protocol $\cP'$ $\alice$ gets input $x$, $\bob$ gets input $j$ (there are no inputs $i,y$) and they compute $x_j$ with slightly larger error than $\epsilon$ (averaged over the uniform distribution on $(x,j)$). To achieve this we choose $i$ in a suitable way and fix it in $\cP$. We then show that choosing $y$ uniformly and independent of $x$ (instead with the correlation $y_i=x$) can cause only small extra error in computing $x_j$ in $\cP$. Hence we get an $\smp$ protocol $\cP'$ for the Index function (with $y$ acting as private randomness of $\bob$). This implies the bound on $\alice$'s message length via Fact~\ref{fact:qrac}, since it is easy to convert an $\smp$ protocol to a one-way protocol. Details follow.

Let the registers $X, I$ hold $\alice$'s inputs, and the registers $Y,J$ hold $\bob$'s inputs. Denote by $E_A,E_B,E_R$ the registers which contain the initial entangled state for $\alice$, $\bob$, and the $\referee$. These registers may hold an arbitrary state independent of the input. Let  register $M_A$ contain $\alice$'s message and register $M_B$ contain $\bob$'s message. 

Let the distribution $\mu$ be such that $y$, $i$ and $j$ are chosen uniformly and independently from their respective domains, and $x = y_i$. Let us put distribution $\mu$ on $(X,I,Y,J)$. 
Now consider the situation when $\bob$ has created his message, but neither $\alice$ nor $\referee$ have done anything yet (this can be assumed since $\alice$ and $\bob$'s operations act on different qubits).
In this situation by Fact \ref{fact:lowinfent} we have $I(JE_AE_RM_B:Y)\leq 2|M_B|$ and hence 
\begin{equation} \label{eq:lowexpect} \E_{i \leftarrow I}[I(JE_AE_RM_B:Y_i)]\leq 2|M_B|/n=2\delta \enspace . \end{equation} 
This can be shown using Fact~\ref{fact:chaininf} since the collection $\{Y_i \; : \; i \in [n]\}$ is independent.

Denote by $\sigma_{i,x}$ the joint state of $J,E_A,E_R,M_B$ when $I=i$ and $X=x$ (and hence $Y_i=x$).
Setting $\sigma_i \defeq \E_{x \leftarrow X} [\sigma_{i,x}]$ we get from Eq.~\ref{eq:lowexpect} and Eq.~\ref{eq:altchar}: $\E_{i \leftarrow I}\E_{x \leftarrow X}[S(\sigma_{i,x}\|\sigma_i)]\leq2\delta$. Let $\ti \in [n]$ be such that $\E_{x \leftarrow X}[S(\sigma_{\ti,x}\|\sigma_{\ti})]\leq2\delta$. Fact \ref{fact:entropytrace} and concavity of the square root function now implies:
\begin{equation} \label{eq:lowtr} \E_{x \leftarrow X}\tn{\sigma_{\ti,x}-\sigma_{\ti}}\le2\sqrt{\delta}  \enspace .
\end{equation} 

Let the distribution $\mu_{\ti}$ be obtained from $\mu$ by fixing $i = \ti$. Let $\rho_{\ti}$ be the joint state of $J,E_A,E_R,M_B,X$, just after $\bob$ has created his message in the protocol $\cP$, when we start with distribution $\mu_{\ti}$ on $(X,I,Y,J)$. Let the distribution $\mu'_{\ti}$ be such that all of $x,y,j$ are chosen uniformly and independently (without any correlation between $y_{\ti}$ and $x$) and $i$ fixed to $\ti$. Let $\theta_{\ti}$ be the joint state of $J,E_A,E_R,M_B,X$, just after $\bob$ has created his message in $\cP$, when we start with distribution $\mu'_{\ti}$ on $(X,I,Y,J)$. 
Note that, using Eq.~\ref{eq:lowtr} we have,
\begin{equation} \label{eq:nice} \tn{\rho_{\ti}-\theta_{\ti}}=\E_{x \leftarrow X}\tn{\sigma_{\ti,x}-\sigma_{\ti}}\le2\sqrt{\delta} \enspace .\end{equation}

Note that the "relevant" registers for correctness of the protocol $\cP$ (after Bob's message is generated) are only $X,J,E_A,E_R,M_B$ (since the output needs to be $X_J$). When we start with distribution $\mu_{\ti}$ on $(X,I,Y,J)$, the protocol $\cP$ would be correct with probability $1-\epsilon$ (since all inputs with positive probability under $\mu_{\ti}$ satisfy the promise $y_i=x$), and changing the state of all "relevant" registers from $\rho_{\ti}$ to $\theta_{\ti}$ can introduce an average extra error of at most $\sqrt\delta$ in computing $X_J$ (due to Eq.~\ref{eq:nice})\footnote{This is a standard fact that follows due to monotonicity of trace distance under admissible quantum operations.}. 

Now consider the protocol $\cP'$ for the Index function in which on inputs $(x,j)$ to $\alice$ and $\bob$ respectively (with $x,j$ drawn uniformly and independently), $\alice$ fixes input $i$ in $\cP$ to $\ti$, $\bob$ generates a $y$ uniformly and independent of $(x,j)$ using private coins, and then $\alice$, $\bob$ and $\referee$ proceed with the rest of the protocol $\cP$. Note that in this case registers $(X,I,Y,J)$  have distribution $\mu'_{\ti}$ on them. Due to our earlier observation, distributional error of $\cP'$, under $\mu'_{\ti}$, is at most $\epsilon+\sqrt\delta$.  Now $\cP'$ can trivially be turned into a one-way quantum protocol $\cP''$ with entanglement between $\alice$ and $\bob$ (by letting $\bob$ do also the role of $\referee$), and $\alice$ sending the message of same length as in $\cP'$. By Fact \ref{fact:qrac}, $\cP''$ needs communication $(1-H(\epsilon+\sqrt\delta))n/2$, hence $\alice$'s message in $\cP'$ must be that long.
\end{proof}

\section{Conclusions and open problems}
\label{sec:conc}

We have shown a tight (up to an additive log factor) Direct sum result for the randomized $\smp$-complexity with private coins, and a tight Direct sum result for the $\sQ^{\|,\tpriv}$ model. While for some relations like one investigated in \cite{GavinskyKRW06} lower bounds known for  $\sQ^{\|,\priv}$ can be extended to the  $\sQ^{\|,\tpriv}$ model, the general relation between those models remains unknown, and is related to the general open question of how useful entanglement is in quantum communication. The main open problems here are, however, to show a Direct Sum result for the $\sQ^{\|,\ent}$ model, and for the  $\sQ^{\|,\priv}$ model, or disprove such statements.

Furthermore we have investigated the gap between the $\smp$ model and the one-way model. We have described an exponential gap between the fully entangled quantum $\smp$ model and the deterministic one-way model for a partial function, which is optimal in the sense that such a gap does not hold for total functions. However, most likely there is an exponential gap between the  $\sQ^{\|,\ent}$ model and randomized one-way complexity for the (total function variant) Generalized Addressing function, but we have only been able to lower bound the $\sQ^{\|,\pub}$ complexity of this problem. Finally, lower bounds for this function in any mode that exceed the $\sqrt n$ barrier, or improved upper bounds would be very interesting.

\bibliography{smp-direct-priv}

\appendix
\section{Proofs of Facts}
\label{sec:proofs}
\begin{proofof}{Fact~\ref{fact:lowinfent}}
We have the following Araki-Lieb~\cite{ArakiL70} inequality for any two systems $M_1,
M_2$: $|S(M_1) - S(M_2)| \leq S(M_1M_2)$. This implies:
$$ I(M_1:M_2) = S(M_1) + S(M_2) - S(M_1M_2) \leq \min\{2S(M_1), 2S(M_2)\} \enspace .$$
Now,
\begin{eqnarray*}
I(X:MN) & = & I(X:M) + I(XM:N) - I(M:N)  \\
& \leq &  I(XM:N) \quad \leq \quad 2S(N) \enspace .
\end{eqnarray*}
\end{proofof}

\vspace{0.2in}

\begin{proofof}{Fact~\ref{fact:supadd}}
We show the fact for $k=2$, which easily implies the same for larger $k$.
Let $\cX \defeq \cX_1$ and $\cY \defeq \cX_2$.
For each $x \in \cX$, let~$E^x : \cY \rightarrow \cS$ be an $E$-derived channel on $\cY$ given by $E^x(y) \defeq E(x,y)$.
For each $x \in \cX$, let ${\mu_x}$ be a
probability distribution on $\cY$ such that $I(E^x_{\mu_x}) = C(E^x)$. Now
let $E^\cX : \cX  \rightarrow \cS$ be an $E$-derived channel on $\cX$ given by $E^\cX(x) \defeq \E_{y \leftarrow \mu_x}[E(x,y)]$.
Let ${\mu_\cX}$ be a distribution on $\cX$ such that $I(E^\cX_{\mu_\cX}) = C(E^\cX)$. Let
$\mu$ be the distribution on $\cX \times \cY$ arising by sampling from
$\cX$ according to $\mu_\cX$, and conditioned on sampling $x$, sampling from $\cY$ according to
$\mu_x$.  Now the following {\em chain rule property} holds for  mutual information.
\begin{fact}
\label{fact:chaininf}
Let $X, Y, Z$ be a tripartite system where $X$ is a classical system. Let $P$ be the distribution of X. Then,
$$I(XY : Z) = I(X : Z) + \E_{x \leftarrow P} [I((Y : Z) \ | \ X =x)] \enspace .$$
\end{fact}
Now we have,
\begin{eqnarray*}
C(E) & \geq & I (E_\mu) \quad \mbox{(from definition of capacity)} \\
& = & I(E^\cX_{\mu_\cX})  + \E_{x \leftarrow \mu_X}[I(E^x_{\mu_x})] \quad \mbox{(from chain rule for mutual information)}\\
& = & C(E^\cX) + \E_{x \leftarrow \mu_X} [C(E^x)]  \\
& \geq & \min_{F_1 \in \cC_1} C(F_1) + \min_{F_2 \in \cC_2} C(F_2) \enspace .
\end{eqnarray*}
This finishes the proof.
\end{proofof}

\vspace{0.2in}

\begin{proofof}{Fact~\ref{fact:SleqC}}
We will need the following {\em joint convexity} property of relative entropy. For
quantum states $\rho_1, \rho_2, \sigma_1, \sigma_2$ and $p \in [0,1]$ we have:
$$ S(p\rho_1 + (1-p)\rho_2 \| p \sigma_1 + (1-p)\sigma_2 ) \quad \leq \quad p \cdot  S(\rho_1 \| \sigma_1) +
(1-p)\cdot  S(\rho_2 \| \sigma_2) \enspace .$$

We will require the following minimax theorem from game theory, which
is a consequence of the Kakutani fixed point theorem in real analysis.
\begin{fact}
\label{fact:minimax}
Let $A_1,A_2$ be non-empty, convex and compact subsets of $\R^n$  ($\R$ stands for the set of real numbers)
for some positive integer $n$. Let $u: A_1 \times A_2 \rightarrow \R$ be a
continuous function, such that
\begin{enumerate}
\item \label{it:first} $\forall a_2 \in A_2$, the set
$\{a_1 \in A_1: u(a_1,a_2) = \max_{a_1' \in A_1} u(a_1',a_2)\}$
is convex; and
\item  \label{it:second} $\forall a_1 \in A_1$, the set
$\{a_2 \in A_2:  u(a_1,a_2) = \min_{a_2' \in A_2} u(a_1,a_2')\}$
is convex.
\end{enumerate}
Then, there is an $(a_1^\ast, a_2^\ast) \in A_1 \times A_2$ such that
\begin{displaymath}
\max_{a_1\in A_1}\, \min_{a_2\in A_2} u(a_1,a_2)
= u(a_1^\ast, a_2^\ast) = \min_{a_2 \in A_2}\, \max_{a_1 \in A_1}
u(a_1,a_2).
\end{displaymath}
\end{fact}
\paragraph{Remark:}
The above statement follows by combining Proposition~20.3 (which shows
the existence of Nash equilibrium $a^\ast$ in strategic games) and
Proposition~22.2 (which connects Nash equilibrium and the min-max
theorem for games defined using a pay-off function such as $u$) of
Osborne and Rubinstein's~\cite[pages 19--22]{osborne:gametheory} book
on game theory.

Let $A_1 = A_2$ be the set of all distributions on the set $\cX$. Since
$\cX$ is finite, $A_1,A_2$ are convex and compact subsets of $\R^n$ for
some $n$.
Let $\rho_x \defeq E(x)$. For distribution $\mu$ on $\cX$, let $\rho_\mu
\defeq \E_{x \leftarrow \mu} [\rho_x]$. Let the function $u: A_1
\times A_2 \mapsto \R$ be such that $u(\lambda,\mu) \defeq \E_{x
\leftarrow \lambda}[S(\rho_x \| \rho_{\mu})]$. The
condition~\ref{it:first} of Fact~\ref{fact:minimax} can be easily seen
to be satisfied since $u(\cdot , \cdot)$ is linear in the first argument. For condition~\ref{it:second} consider the
following. Fix $\lambda \in A_1$. Let $\mu_1, \mu_2 \in A_2$ be such
that $u(\lambda, \mu_1) = u (\lambda, \mu_2) = \min_{\mu'} u(\lambda,
\mu')$. Let $p \in [0,1]$; we need to show that $\mu_p \defeq
p\mu_1 + (1-p) \mu_2$ satisfies $u (\lambda, \mu_p) = \min_{\mu'} u(\lambda,
\mu')$. We have from joint
convexity of relative entropy:
\begin{eqnarray*}
\E_{x \leftarrow \lambda} [S(\rho_x \| \rho_{\mu_p})] & \leq & \E_{x \leftarrow \lambda}[p \cdot  S(\rho_x \| \rho_{\mu_1}) +
(1-p)\cdot  S(\rho_x \| \rho_{\mu_2})] \\
& = & p \cdot \E_{x \leftarrow \lambda} [S(\rho_x \| \rho_{\mu_1})] +
(1-p)\cdot  \E_{x \leftarrow \lambda} [S(\rho_x \| \rho_{\mu_2})] \\
& = & p \cdot u(\lambda, \mu_1) + (1-p) \cdot u(\lambda, \mu_2) =
\min_{\mu'} u (\lambda, \mu') \enspace .
\end{eqnarray*}

Therefore we have:
\begin{eqnarray*}
 \min_{\mu} \max_x S(\rho_x|| \rho_{\mu}) & = &  \min_{\mu}
\max_\lambda \E_{x \leftarrow \lambda} [S(\rho_x|| \rho_{\mu})] \\
& = & \min_{\mu} \max_\lambda u(\lambda, \mu) \\
& = & \max_{\lambda}\min_{\mu} u(\lambda, \mu) \quad \mbox{(from Fact~\ref{fact:minimax})}\\
& = & \max_{\lambda}\min_{\mu} \E_{x \leftarrow \lambda}[S(\rho_x || \rho_{\mu})] \\
& \leq & \max_{\lambda} \E_{x \leftarrow \lambda} [ S(\rho_x|| \rho_{\lambda})]  \\
& = & \max_{\lambda}  I(E_{\lambda}) = C(E)
\end{eqnarray*}
Therefore there exists $\tilde{\mu} \in A_2$ such that $\max_{x \in \cX} S(\rho_x|| \rho_{\tilde{\mu}}) \leq C(E)$. We let $\tau \defeq \rho_{\tilde{\mu}}$ and conclude our proof.
\end{proofof}

\vspace{0.2in}
\begin{proofof}{Fact~\ref{fact:compress}}
We use the following information-theoretic result called
the {\em substate theorem} due to Jain, Radhakrishnan, and
Sen~\cite{JainRS02}.
\begin{fact}[Substate theorem~\cite{JainRS02}]
\label{fact:substate}
Let $\cH, \cK$ be two finite dimensional Hilbert spaces and \newline
$\dim(\cK) \geq \dim(\cH)$. Let $\C^2$ denote the two dimensional
complex Hilbert space.  Let $\rho, \tau$ be density matrices in
$\cH$ such that $S(\rho \| \tau) < \infty$.  Let
$\ket{\overline{\rho}}$ be a
purification of $\rho$ in $\cH \otimes \cK$. Then, for $r > 1$, there
exist pure states $\ket{\psi}, \ket{\theta} \in \cH
\otimes \cK$ and $\ket{\overline{\tau}} \in \cH \otimes \cK \otimes \C^2$,
depending on $r$, such
that $\ket{\overline{\tau}}$ is a purification of $\tau$ and
$\tn{\ketbra{\overline{\rho}} - \ketbra{\psi}} \leq \frac{2}{\sqrt{r}}$, where
\begin{displaymath}
\ket{\overline{\tau}} \defeq
\sqrt{\frac{r-1}{r 2^{r k}}} \, \ket{\psi}\ket{1} +
\sqrt{1 - \frac{r-1}{r 2^{r k}}} \, \ket{\theta}\ket{0}
\end{displaymath}
and $k \defeq 8 S(\rho \| \tau) + 14$.
\end{fact}
We will also require the following fact that is easily shown using {\em Schmidt decompositions} of pure states.
\begin{fact}[Local-transition]
Let $\rho$ be a quantum state in $\cK$. Let $\ket{\phi_1}$ and $
\ket{\phi_2}$ be two purification of $\rho$ in $\cH \otimes \cK$. Then
there is a local unitary transformation $U$ acting on $\cH$ such that
$(U \otimes I) \ket{\phi_1} = \ket{\phi_2}$.
\end{fact}

Let $c \defeq S(\rho \| \tau)$. Let us invoke  Fact~\ref{fact:substate} with $\ket{\overline{\rho}}$
being any purification of $\rho$ and $r \defeq 16/\delta^2$. Let $\ket{\overline{\tau}}$ be the purification of
$\tau$ as given by Fact~\ref{fact:substate}. Let $\alice$ and $\referee$ start with $2^{\frac{2rk}{\delta}}$ ($k \defeq 8c + 14$)
copies of the pure state $\ket{\phi}$, such that marginal of $\ket{\phi}$ on $\referee$'s side is $\tau$.
Since the reduced quantum state on $\referee$'s part in both $\ket{\psi}$ and $\ket{\overline{\tau}}$ is the
same, from local-transition fact, there exists a transformation
acting only in $\alice$'s side which takes $\ket{\phi}$ to
$\ket{\overline{\tau}}$. $\alice$ transforms each $\ket{\psi}$ to
$\ket{\overline{\tau}}$ and measures the first bit. If she obtains 1 in any
copy of $\ket{\overline{\tau}}$ she communicates the number of that copy to
$\referee$. In case she fails to obtain $1$ in $2^{\frac{2rk}{\delta}}$ trials, she communicates this to $\referee$ and $\referee$ assumes the state $\ket{0}\bra{0}$.
It is easily seen that the communication from $\alice$ is
at most $O(\frac{c}{\delta^3})$. Also since $\Pr(\mbox{$\alice$ observes $1$}) = \frac{r-1}{r 2^{r
k}}$, and $\alice$ makes $2^{\frac{2 r k}{\delta}}$ tries she succeeds with probability at least $1 - \delta/2$. In case she succeeds, let the
state with $\referee$ in which $\alice$ succeeds be $\tilde{\rho}$.  From
Fact~\ref{fact:substate} and monotonicity of trace-norm, $\tn{\tilde{\rho} -  \rho} \leq \delta/2$.
So for the final state $\rho'$ produced with $\referee$, it follows that $\tn{\rho' - \rho} \leq \delta$.
\end{proofof}

\vspace{0.2in}
\begin{proofof}{Fact~\ref{fact:compressclass}}
This proof follows on very similar lines as that of Fact~\ref{fact:compress}.
We use the following classical substate theorem~\cite{JainRS02}.
\begin{fact}[Classical substate theorem]
\label{fact:substateclass}
Let $P, Q$ be probability distributions on the same set such that $S(P \| Q) < \infty$.  For every $r > 1$, there exist distributions $\tilde{P}, R$  such that
$\|P- P'\| \leq 2/r$ and $Q = \frac{r-1}{r 2^{r k}} \tilde{P} + (1 - \frac{r-1}{r 2^{r k}})R $, where $k \defeq  S(\rho \| \tau) + 1$.
\end{fact}

We will also need the following easily verifiable fact.

\begin{fact}
\label{fact:class}
Let $X$ be a random variable distributed according to $Q$. Let $p \in [0,1]$ and $Q_1, Q_2$ be distributions such that $ Q = pQ_1 + (1-p)Q_2$. There exists a binary random variable $Z \in \{0,1\}$, correlated with $X$, with $\Pr[Z=1] = p$, such that the distribution of $X$ conditioned on $Z=1$ is $Q_1$ and  the distribution of $X$  conditioned on $Z=0$ is $Q_2$.
\end{fact}

Let $c \defeq S(P \| Q)$. Let us invoke Fact~\ref{fact:substateclass} with $r \defeq 4/\delta$ and let $\tilde{P}, R$ be as obtained by Fact~\ref{fact:substateclass}. Let $X$ be a random variable distributed according to $Q$.
Let $\alice$ and $\referee$ share $2^{\frac{2rk}{\delta}}$ ($k \defeq c + 1$) copies of $X$ as public randomness. Let $Z$ be a random variable, correlated with $X$, obtained from Fact~\ref{fact:class} by letting $Q_1 \defeq \tilde{P}$ and $Q_2  \defeq R$. $\alice$ generates the random variable $Z$ for each copy of $X$, measures $Z$ and sends the number of the first copy in which she succeeds to obtain $1$ to $\referee$.   In case she fails to obtain a $1$ in $2^{\frac{2rk}{\delta}}$ trials, she communicates this to $\referee$ and $\referee$ assumes single point distribution concentrated on $0$.
It is easily seen that the communication from $\alice$ is
at most $O(\frac{c}{\delta^2})$. Also since $\Pr(\mbox{$\alice$ observes $1$}) = \frac{r-1}{r 2^{r
k}}$, and $\alice$ makes $2^{\frac{2 r k}{\delta}}$ tries she succeeds with probability at least $1 - \delta/2$. In case she succeeds, the copy which she communicates to $\referee$ will be distributed according to $\tilde{P}$.
 From Fact~\ref{fact:substate}, $\|\tilde{P} -  P\| \leq \delta/2$.
So for the final distribution $P'$ produced with $\referee$, it follows that $\|P' - P\| \leq \delta$.
\end{proofof}

\end{document}